\documentclass{article}
\usepackage{ijcai17}
 
\newcount\Comments  
\newcommand{\kibitz}[2]{\ifnum\Comments=1{\color{#1}{#2}}\fi}
\newcommand{\hma}[1]{\kibitz{blue}{[HM: #1]}}
\newcommand{\rmr}[1]{\kibitz{red}{[RESHEF: #1]}}
\newcommand{\vrb}[1]{\kibitz{cyan}{[Valentin: #1]}}

\newcount\PGE  
\newcommand{\PGEcommand}[2]{\ifnum\PGE=1{\color{#1}{#2}}\fi}

\newcount\Version  
\newcommand{\ver}[2]{\ifnum\Version=1{#1}\fi \ifnum\Version=2{#2}\fi}

\def\newpar{\vspace{-0mm}\paragraph}

\newcommand{\calS}{\mathcal{S}}

\def\ol{\overline}

\usepackage{times}
\usepackage{graphicx}
\usepackage{pgfplots}
\usepackage{enumitem}
\usepackage{etex}
\usepackage{latexsym}
\usepackage{amssymb}
\usepackage{amsmath}
\usepackage{amsfonts}
\usepackage{bbm}
\usepackage{ifthen}
\usepackage{psfrag}
\usepackage{floatflt}
\usepackage{verbatim}
\usepackage{scalefnt}
\usepackage{epsfig}
\usepackage{epstopdf}
\usepackage{color} 
\usepackage{xcolor}
\usepackage{calc}
\usepackage{graphicx} 
\usepackage{subfig}
\usepackage{multicol}
\usepackage{tikz}
\usepackage{amsthm}	

\usepackage{url}
\usetikzlibrary{arrows,decorations,decorations.shapes,decorations.markings,decorations.pathreplacing,backgrounds,shapes,petri,topaths}

\newtheorem{theorem}{Theorem}

\newtheorem{proposition}[theorem]{Proposition}

\theoremstyle{definition}
\newtheorem{exmp}{Example}
\newenvironment{rtheorem}[1]{\medskip\noindent\textbf{Theorem~\ref{#1}.}\begin{itshape}}{\end{itshape}}

\newenvironment{rproposition}[1]{\medskip\noindent\textbf{Proposition~\ref{#1}.}\begin{itshape}}{\end{itshape}}

\DeclareMathOperator{\argmin}{argmin}
\DeclareMathOperator{\argmax}{argmax}

\title{Contract Design for Energy Demand Response}

\author{Reshef Meir\\
Technion -- Israel Institute of Technology\\
Haifa, Israel\\
reshefm@ie.technion.ac.il\\
\And
Hongyao Ma\\
Harvard University\\
Cambridge, MA, US\\
hma@seas.harvard.edu\\
\And
Valentin Robu\\
Heriot-Watt University\\
Edinburgh, UK\\
V.Robu@hw.ac.uk\\
}

\begin{document}

\maketitle

\begin{abstract}
Power companies such as Southern California Edison (SCE) uses Demand Response (DR) contracts to incentivize consumers to reduce their power consumption during periods when demand forecast exceeds supply. 
Current mechanisms in use offer contracts to consumers independent of one another, do not take into consideration consumers' heterogeneity in consumption profile or reliability, and fail to achieve high participation. 

We introduce DR-VCG, a new DR mechanism that offers a flexible set of contracts (which may include the standard SCE contracts) and uses VCG pricing. We prove that DR-VCG elicits truthful bids, incentivizes honest preparation efforts, enables efficient computation of allocation and prices. With simple fixed-penalty contracts, the optimization goal of the mechanism is an upper bound on probability that the reduction target is missed.\hma{Reviewer 1 asked us to clarify reliability guarantees in the abstract}\rmr{ok?} Extensive simulations show that compared to the current mechanism deployed in by SCE, the DR-VCG mechanism achieves higher participation,  increased reliability, and significantly reduced total expenses. 

\end{abstract}

\section{Introduction}
Power system operation involves many challenges, driven by the requirement that supply equals demand at all times. Too much supply may lead to overload on the grid, whereas excessive demand may lead to shortages and blackouts. The problem is aggravated by the fact that consumption tends to vary sharply due to certain events (for example, surges in consumption during heatwaves), whereas increasing the supply is typically slow and costly. Even if the power company wants to shift some of the demand to a different time, it cannot coerce the consumers to do so, and may only affect their behavior by using monetary incentives, such as increasing electricity price during peak-demand times. As in other markets, consumers  may respond to incentives in different ways based on their own preferences. 
Unlike some other markets, the serious consequences of failure to meet demand, and the large uncertainty about how consumers may react to incentives, requires the power company to guarantee there is enough slack on the supply side.


 \hma{This feels like a lot of words to express the idea that market failure has big consequences. I think the point that we need to make here is ``since market failure has huge consequences, sufficient slack need to maintained in the system (i.e. reserves on the supply side) which is very costly. And we also need to mention that there is large uncertainty in the demand side's ability to reduce consumption. (It's not obvious from the abstract or the rest of the intro, and otherwise the problem would be trivial.)} \rmr{better?}

\newpar{The DR-SCE mechanism}
Demand Response (DR) programs are used by power companies to handle surges in demand by reducing consumption rather than increasing production. Typically, when a surge is predicted one day ahead of time, the company lets consumers bid on how much consumption they can reduce. \hma{elicits bids from consumers on the amount of consumption they plan on reducing?} Each consumer is being paid a fixed $\$0.5$ per reduced kWh, but only if the reduced consumption is between 50\% and 150\% of her bid. We call this system the \emph{DR-SCE mechanism}, as it is used by  Southern California Edison (c.f.~\cite{DBP_eval}) (as well as other companies such as PG\&E~\cite{DBP_report}). \hma{add parenthesis around ``as well as ..."? since we were describing why it's named SCE?}

DR-SCE has several shortcomings: 
incentives for participation are often insufficient (only 12\% of registered participants in 2012-2013 submitted any bids)\hma{is there a number on the percentage of consumers that registered?}, the system does not capture the very different consumption profiles of consumers, and does not filter out unreliable \hma{again I think uncertainty should be mentioned earlier, otherwise it's not clear what ``unreliable" means here} bidders. Yet, being a widely deployed DR system, we treat DR-SCE both as a starting point and as a benchmark for new mechanisms.  \hma{the different consumption profiles are mentioned many times. I'm just wondering if it would be good to mention a couple of short examples to show some agents are inherently less reliable than some others?}\rmr{a good idea but no time or space}

\rmr{Based on the evaluation report, the main problem is engagement: only 12\% of target clients actually bid. Participants said that it was not ``worth the effort`` and that ``preparation is too costly'' (page 30-31, 36). Our mechanism allows bidders to demand compensation based on their true effort (and if it is too much, they will simply not be selected). There is also a difference between ``sophisticated'' agents (who tend to be active and monitor their own actions) and ``simple'' agents who are either inactive or just bid the default.  }

\newpar{Contribution} 
We propose a novel DR-VCG mechanism for selecting and incentivizing a subset of consumers to reduce consumption. The grid offers a set of contracts defined by some desired reduction target and a penalty scheme, and agents may bid how much they want to get paid on for accepting each contract. The mechanism then selects a subset of contracts that minimizes the \emph{sum of bids}, and applies VCG prices to pay the agents. As a result, it is a dominant strategy for all agents to bid their true costs. 

We show that for natural penalty schemes, the sum of bids is a good proxy for the reliability of the joint contract, as high bids are indicative of low individual reliability. We show that the current contracts \hma{based on fixed per unit reward?} used by SCE and PG\&E can still be offered under DR-VCG (to allow for easy transition and backward-compatibility). We demonstrate via examples and simulations that even when restricted to offering SCE-like contracts, DR-VCG dominates DR-SCE in terms of reliability and grid expense. \rmr{removed: We also highlight the advantage of using more general \emph{interval contracts} which are simple and intuitive, enable efficient computation of VCG prices, and capture other considerations important to the grid.}
\ver{

This is the full version of a paper accepted to IJCAI-2017~\cite{MMR17}. All omitted proofs are available on Appendix~\ref{apx:proofs}.}{All omitted proofs are available in the full version of this paper~\cite{MMR17_full}.}

\newpar{Related Work}
A number of recent works have discussed how groups of agents can be coordinated and incentivized to shift power demand~\cite{haring_al,Zhang_al,Su_al}. Considering strategic agents, ~\cite{rose_al} and \cite{AkasiadisC13} propose the use of {\em scoring rules} to incentivize truthful reports about expected future generation or consumption. However, scoring rule approaches are not concerned with selection of agents to satisfy a system-wide reliability constraint. 
\cite{NaLi_al} consider agents bidding using supply curves, and study the market equilibria for this setting. They do not, however take a mechanism design perspective or try to guarantee truthfulness. 
Mechanism design approaches for aggregating load apply either variations of VCG~\cite{samadi2012advanced,chapman_verbic}, or a new  ``staggered clock-proxy" auction~\cite{nekouei_al}. Neither work considers the crucial issue of \emph{reliability}, i.e. in practice not all agents selected to respond will do so.
 A different variation of VCG pricing was used by \cite{porter2008fault} to align the incentives of agent in face of possible failures in general mechanisms. However their version requires full revelation (which is problematic in practical DR programs), and is aimed at maximizing social welfare rather than reliability.  \rmr{added}

The work closest to ours is \cite{Ma_ijcai16}, who propose a mechanism that allows agents to bid on the maximal penalty for failing the DR contract, while \cite{ma_AAMAS17} extend this work to include uncertainty about costs. Our work takes a different approach, which generalizes currently used contracts and, in our view, is more geared to practical applications. \rmr{Hongyao want to edit this?}

\section{Model} 

We consider a single power utility or system operator (henceforth, the \emph{grid}), and a set $N$ of consumers (or \emph{agents}) who are registered to a demand response program.
\hma{This would probably be too big of a change at this moment, but I feel that it would be  easier for people to understand if we talk about the multi-period model, agents' types and beliefs etc earlier when talking about the model, so when people read about the contracts they would be thinking about what agents' preferences / behavior would be under each of them. How about briefly mentioning that each agent can choose from a few levels of preparation efforts a day ahead of time at certain costs, and when the time comes the randomness in her amount of reduction is realized?}

The most important task of the grid is to cut down consumption by at least $M$ energy units (say, $kWh$) during the DR event. Given that this target is met, the grid   would like to minimize payments to the agents. 
The baseline consumption profile of each agent is assumed known from past consumption data, so the grid can measure how much each agent reduced in practice. 

\subsection{Contracts}
A contract in our model is defined by a pair $(\ell,F)$, where $\ell$ is a \emph{commitment goal} in energy units and  $F: \mathbb N \rightarrow \mathbb R_+$ is a penalty function, mapping the realized energy reduction $X$ to a monetary penalty $F(X)$.
A-priori, $F$ is unconstrained, and $\ell$ is merely a non-binding declaration of the agent's intentions. It makes sense to consider more specific classes of penalties that attach the penalty to the commitment goal.
\begin{description}[leftmargin=0.3cm]
	\item[Fixed contracts] A \emph{Fixed penalty contract} \rmr{Capitalize Fixed, Cliff, etc.?} is defined by a pair $(\ell,f_\ell)$, and the penalty is set to $F(X)=f_\ell$ if $X<\ell$ and 0 otherwise. In other words, the agent commits to reduce $\ell$, or otherwise pay a penalty of $f_\ell$.  See Fig.~\ref{fig:contracts}(a).
	\item[Cliff contracts] A \emph{Cliff penalty contract} is defined by a tuple $(\ell,f_\ell,\alpha,\beta)$ where $\alpha<1, \beta>0,f_\ell\geq \ell(1-\alpha)\beta$. It has the following form:
	
	\vspace{-7mm}
	$$~~~~~~~~~~~~~~~F(X) = \left\{\begin{array}{ll}
f_\ell, &X<\alpha\cdot \ell \\
(\ell-X)\beta,  & \alpha\cdot \ell  \leq X < \ell \\ 
0 ,& \ell \leq X
\end{array}
\right.$$

We can think of a Cliff penalty function as a plateau where the penalty is 0 whenever the commitment $\ell$ is met. Failure to meet the goal results in a linear penalty, where beyond a certain point the penalty becomes a  constant, and the utility drops sharply (hence a ``cliff''). See Fig.~\ref{fig:contracts}(b).
\end{description}

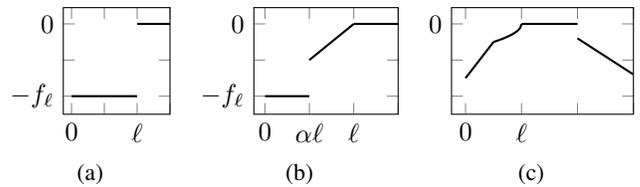
\begin{figure}[t] 
\centering
\subfloat[]{ 
\begin{tikzpicture}[scale=1]
\begin{axis}[width=3cm,height=3cm, enlarge x limits = -1,xticklabels={,$0$,,$\ell$,,},yticklabels={,$-f_\ell$,,$0$},xmin=-0.5,xmax=6,ymin=-5,ymax=0.9]
    \addplot[thick,domain=0:4] {-4}; 
		  \addplot[thick,domain=4:6] {0};
		\end{axis}
\end{tikzpicture}
}
\subfloat[]{ 
\begin{tikzpicture}[scale=1]
\begin{axis}[width=3.5cm,height=3cm, enlarge x limits = -1,xticklabels={,$0$,$\alpha \ell$,$\ell$,,},yticklabels={,$-f_\ell$,,$0$},xmin=-0.5,xmax=6,ymin=-5,ymax=0.9]
    \addplot[thick,domain=0:2] {-4}; 
		  \addplot[thick,domain=2:4] {x-4 };
			\addplot[thick,domain=4:6] {0};
		\end{axis}
\end{tikzpicture}
}
\subfloat[]{
\begin{tikzpicture}[scale=1]
\begin{axis}[width=4cm,height=3cm, enlarge x limits = -1,xticklabels={,$0$,$\ell$,,,},yticklabels={,,,$0$},xmin=-0.5,xmax=6,ymin=-5,ymax=0.9]
    \addplot[thick,domain=0:1] {-1-(1-x)*2}; 
		\addplot[thick,domain=1:2] {-sqrt(2-x)}; 
		  \addplot[thick,domain=2:4] {0};
			\addplot[thick,domain=4:6] {-0.8-(x-4)};
		\end{axis}
\end{tikzpicture}
}\vspace{-2mm}
\caption{The penalty $-F$ as a function of the realized reduction $X$ under a Fixed contract (a), a Cliff contract (b) and a general contract (c).\label{fig:contracts}\vrb{It is not immediately clear from this figure that case (c) is a generalization of (b), as claimed in the text. May be a trivial point, but the lower cliff penalty level (up to $\alpha\ell$) is missing in (c).}\vspace{-4mm}}
\end{figure}


Clearly any Fixed contract is a Cliff contract. 
As we will later see, the SCE payment scheme can be implemented as a particular Cliff penalty scheme, so even restricting our mechanism to using Cliff contracts is sufficient to generalize the SCE system. 

  %
	%
	%

	\newpar{Optimal contract sets}
	For what follows, we assume no structural restriction on contracts or penalty schemes. We simply assume that a set of $k$ contracts $J$ are offered, and $F(j,X)$ is the penalty for an agent who signs up for contract $j\in J$ and reduces consumption by $X$. Since the value of a contract to an agent is always non-positive, denote by $B_{ij}\geq 0$ the bid of agent $i$ on contract $j$.\footnote{The bid is supposed to reflect the various costs involved in taking the contract: preparation cost, online adjustment costs, the expected penalty and so on. We elaborate on this in the next section.} 
 Then for a subset of contracts $S\subseteq N \times J$, we denote the \emph{sum of bids} by $SB(S) = \sum_{(i,j)\in S}B_{ij}$. 
	
In addition, the grid may pose a restriction on which sets of contracts are valid: we denote by $\calS$ all valid sets of contracts. In this work, we use this constraint to impose a lower bound on  (declared) reduction in consumption, thus 
$$\vspace{-2mm}\calS(M) = \{S : \sum_{(i,j)\in S}\ell_j \geq M \text{ and } \forall i |\{(i,j)\in S\}|\leq 1 \}.$$
 In other words, $\calS(M)$ includes all sets of contracts that \emph{claim} to reduce at least $M$ units of consumption, and each agent has at most one contract. For an agent $i\in N$, we denote $SB_{-i}(S) = \sum_{(i',j)\in S : i'\neq i}B_{i'j}$, that is, the sum of bids over all agents in $S$ except $i$.    We sometimes denote $i\in S$ as a selected agent (meaning there is some $j$ s.t. $(i,j)\in S$).  

An \emph{optimal contract set} is a set of individual contracts that minimizes the sum of bids, i.e. $\argmin_{S\in \calS}SB(S)$.

%
 %
%
%
%
%
%
%
\subsection{The DR-VCG Mechanism}

We define the \emph{DR-VCG mechanism}, for assigning demand response contracts using Vickrey-Clarke-Groves (VCG) payments.   
The grid publishes a finite set of contracts $J$.\footnote{The analysis works also for a variant of the mechanism where agents may propose new contracts. \ver{See Appendix~\ref{apx:variants}.}{}} 
  Each agent $i$ submits a single bid $B_{ij}$ on each contract $j$. For now, we can think of $B_{ij}$ as some proxy of the cost required from $i$ when taking on contract $j$.
The mechanism finds the optimal valid set of contracts by solving $\min_{S \in \calS }  SB(S)$. 
	
For a set of agents $N$, a set of contracts $J$, and a reduction goal $M$, we can plug in our more specific optimization goal and constraints.  We get the optimal subset of contracts as
{\small $$\vspace{-2mm}
	S^*(N,J,M) = \argmin_{S \in \calS(M) }\{  \sum_{(i,j) \in S}  B_{ij} \}.
$$
}

We denote $SB^*(N,J,M) = SB(S^*(N,J,M))$, and omit some of the parameters when they are clear from the context. 
Then, for each selected agent $i\in S^*$, the individual rewards are computed as the VCG payments with the Clarke pivot rule~\cite{clarke1971}. Informally, the sum of bids is analog of the \emph{social cost}, and the VCG payment is the positive externality the agent's presence has on the rest of the agents. Formally, for each $i\in N$,
	$$r_i=SB^*(N_{-i},J) - SB^*_{-i}(N,J).$$ 
	The reward is paid to the agent up front, regardless of how much reduction it eventually achieves in practice.  
	
	Finally, for each $(i,j)\in S^*(N,J)$, the selected agent $i$ pays $F(j,X_i)$ to the grid, where each $X_i$ is the realized reduction of agent $i$.
	The utility of agent $i$ depends on the reward, the penalty, and the investment costs required to meet the contract. We analyze agents' incentives and utilities in the next section.

\rmr{here we assume that selected agents are always prompted to reduce consumption. If this is not the case, we only need to suppose that agents know they will be prompted with some fixed probability $q$ (if selected).}
\vrb{This is fine, we make the same assumption as in the other paper, as this is a public probability, known by all.}

\begin{exmp}[Running example] \label{ex:3_bidders}
Suppose we apply the DR-VCG mechanism with a single fixed contract $(\ell=100,f_\ell=50)$. There are three agents that submit bids of $B_1=0,B_2=5,B_3=15$, and the goal for the grid is set to $M=200$. 

The optimal set of contracts that meets $M=200$ is $S^*=\{1,2\}$ with $SB^*=5$. The rewards are:
\vspace{-1mm}
{\small$$r_1 = SB^*(\!\{2,3\}\!)-SB^*_{-1}(\!\{1,2\}\!)\!=\!20-5\!=\!15;\ r_2\!=\!15-0\!=\!15,\vspace{-1mm}$$}
 so in total the grid pays $30$ (some of which it might get back as penalties).  
\end{exmp}  

\hma{I don't think this example adds much to the paper since people know how to compute the VCG payments anyway? The computation can be done with little additional space in Example 2. I also feel it would be nice to write down $S^\ast(N \backslash{1}) = 20$ and $S^\ast(N \backslash{2}) = 15$ if we are going to do the computation?} \rmr{not all people (or reviewers) are like you...} \hma{okay, then there is even more incentive to mention $S^\ast(N \backslash \{1\}) = 20$ and $S^\ast(N \backslash \{2 \}) = 15$?}

\section{Analysis}
\paragraph{Complexity}
In order to run the DR-VCG mechanism, we should be able to efficiently compute the optimal contract set and prices.
Suppose that energy units (including the reduction goal $M$) are integers. 
\begin{theorem}\label{th:complex}For any sets of agents $N$ and Cliff contracts $J$, both of $S^*(N,J)$ and $SB^*(N,J)$  (and thus also VCG prices) can be computed in time polynomial in $n,k,M$. 
\end{theorem}
 In the general case, finding an optimal set of contracts is NP-hard even for fixed contracts, by a  reduction from the Knapsack problem~\cite{karp1972reducibility} (details omitted).

However, the knapsack problem is solvable by dynamic programming when the units are \emph{bounded integers}, and a similar algorithm can be applied to our problem (intuitively, compute dynamically the optimal contract sets for agents $\{1,\ldots,i\}$ for $i=1,2,\ldots,n$). Hence we get Theorem~\ref{th:complex}.

%


\subsection{Incentives}
To make things concrete we will describe a particular probabilistic model from which we can derive agents' costs, and will show how under this model the incentives of the agents align nicely with those of the grid.  \hma{hmm, I guess we would like to convey our contracts align the incentives of agents with that of the grid? In general agents don't want to work and want to get paid so incentives are pretty much mis-aligned...?} 
However the claim the agents' dominant strategy is to reveal their true costs does not depend on this interpretation, and holds whenever agents can attribute a well-defined  cost to each contract. 
 
\newpar{Effort and types}
In general, agents do not know with certainty the amount of energy they will be able to reduce, as this depends on some unknown factors such as urgent service orders, last minute clients and so on. Moreover, preparation may have some cost (e.g. due to changing the work schedule, or turning down orders). By investing a higher effort/cost, an agent might be able to commit to saving more energy. 

\hma{Hmm, again, I think the discussion on uncertainty and types should either be mentioned in the intro or earlier on in e.g. Section 2...}

The \emph{type} of each agent $i$ is given by a distribution $p_i$. In detail,  $p_i(c,X)$ is the probability that by investing $c$, agent $i$ will reduce exactly $X$ units of consumption (thus $\sum_{X\geq 0} p_i(c,X)=1$ for all $c$).
\hma{I'm wondering if it would be good to mention that $c$ can take infinite number of real values, or we an also assume that for each agent there is a set of finite number of $c$'s that the agent can choose from...}\rmr{allow negative $X$? or just assume $X=0$ whenever demand is $\geq$ usual demand?}\hma{This is non-negative by definition? ``baseline consumption profile"  was mentioned in the paragraph before 2.1 Contracts and seems higher than baseline consumption was not allowed?}\footnote{Investment $c$ may include preparation costs,  on-line actions required to produce the energy cut $X$, opportunity cost, and so on.} 
We assume  agents are always trying to maximize their utility.  

A straightforward approach would be to ask agents to report their types (and then apply some version of VCG). However, a language to report an arbitrary distribution may be very complicated. 
Further, unsophisticated agents like small households may not  know their own distribution (or even \emph{what is} a distribution). \hma{hmm could we cut the last piece... I  know people don't understand probability though.} 
Fortunately, our DR-VCG does not require the agents to report any such distribution. 

\hma{ Rewrote the following paragraph.}

Fix an agent $i$ who accepted a contract $j$ (with penalty scheme $F$) and gets paid reward $r_i$. If she decides to invest $c$, she will pay an expected penalty of $EF_i(j,c) = E_{X \sim p_i(c)}[F(j,X)]$, and her  expected utility would be:
$$
	u_i(j,c)  =  r_i - c - EF_i(j,c_i)=r_i - c -E_{X\sim p_i(c)}[F(j,X)]. \vspace{-1mm}
$$
Therefore, the \emph{optimal investment} an utility maximizing agent should make for contract $j$ should be
$$
	c^*_i(j)= \argmax u_i(j,c) = \argmin_{c\geq 0}(c+EF_i(j,c)). \vspace{-1mm}
$$
In words, when agent $i$ is signed up for contract $j$, investing $c^*_i(j)$ will minimize her total cost (investment + penalty). We denote this cost by $C^*_i(j)= c^*_i(j) + EF_i(j,c^*_i(j)) = \min_{c\geq 0}(c+EF_i(j,c))$.
We refer to $C^*_i:J \rightarrow \mathbb R^+$ as the \emph{cost type} of agent $i$, which is derived from her type $p_i$ and $F$. 

\rmr{old paragraph:
Given a penalty scheme $F$, 
If agent $i$ with contract $j$ decides to invest $c$, the expected penalty it is going to pay would be $EF_i(j,c)=E_{X\sim p_i(c)}[F(j,X)]$. 
Thus for any pair $(i,j)$ we can attach an \emph{optimal investment} $c^*_i(j)=\argmin_{c\geq 0}(c+EF_i(j,c))\geq 0$. \hma{``an rational agent would decide on investing ..."?} In words, when agent $i$ is signed up for contract $j$, investing $c^*_i(j)$ will minimize her total cost (investment + penalty). We denote this cost by $C^*_i(j)= c^*_i(j) + EF_i(j,c^*_i(j)) = \min_{c\geq 0}(c+EF_i(j,c))$.
We refer to $C^*_i:J \rightarrow \mathbb R^+$ as the \emph{cost type} of agent $i$, which is derived from her type and $F$. 

The \emph{expected utility} of agent $i$ of type $p_i$ from accepting contract $j$ for reward $r_i$, and  then investing $c_i$ is 
$$u_i(j,c_i)  =  r_i - c_i - EF_i(j,c_i)=r_i - c_i -E_{X\sim p_i(c_i)}[F(j,X)].$$  }

The \emph{total expense} (TE) of the mechanism can be computed as the sum of rewards paid to the agents minus expected penalties\hma{collected from them?}: $TE(S) = \sum_{(i,j)\in S} (r_i - EF_i(j,c_i^*(j))$  (assuming agents invest optimally). 

A mechanism is \emph{truthful} if it is a dominant strategy for any agent $i$ to report her true cost $C^*_i(j)$ on any contract $j$. \rmr{the exact set of contracts depends on the variant of the mechanism} 

A mechanism is \emph{individually rational} (IR) if for every pair  $(i,j)$ selected by the mechanism, $u_i(j,c^*_i(j))\geq 0$  (that is, by participating each agent does not lose in expectation). 

\begin{theorem}\label{th:auction}
Consider  DR-VCG  with arbitrary $N,J,M$.
\begin{enumerate}
	\item For every contract $j\in J$, it is a dominant strategy for agent $i$ to bid   $B_{ij}=C^*_{i}({j})$;
	\item If contract $(i,j)$ is selected, it is a dominant strategy for $i$ to invest $c^*_{i}(j)$; 
	\item The mechanism is IR.	
\end{enumerate}
\end{theorem}
\begin{proof}[Proof sketch] 
Intuitively, we show that VCG payments are market clearing (following similar proofs in other domains, see \cite{nisan2007introduction}), i.e. that no agent prefers a different contract (or no contract) under the given prices \hma{this sounds like agent-independence instead of market clearing?}\rmr{it's all the same}. Since the prices that agent~$i$ faces are independent of her bids, it is a dominant strategy to report truthfully. 
Once  contract $(i,j)$ is selected, then $u_i(j,c)=r_i-c-EF_i({j},c)$.  By definition, this is maximized by investing $c^*_i({j})$.
\end{proof}

\begin{exmp}\label{ex:3_bidders_VCG}
Consider three agents, where each one can reduce consumption by $100$ kWh without effort. However, agents have different reliability and only manage to hold their commitment with a probability $p_i$ of $1, 0.9,$ and $0.7$, respectively (and otherwise reduce $0$). Suppose that the goal of the grid is  $M=200$ kWh. 

Consider the DR-VCG mechanism with a single fixed contract $(\ell=100,f_\ell=50)$. Since agent~1 always meets her commitment, $C^*_1=0$. For the others, $C^*_2=0.1\cdot 50 =5$ (as agent~2 fails w.p.~$0.1$),  and $C^*_3=0.3\cdot 50 =15$, i.e. $B_i=C^*_i$ are exactly the bids in Ex.~\ref{ex:3_bidders}. The selected set is $S^*=\{1,2\}$  (the DR-VCG mechanism filters out the least reliable agent). 
The total expense of the grid is $TE^{VCG}=30-5=25$ (expected penalty of $5$ from agent 2). 
\end{exmp}

\subsection{Fallback options and reserve costs}
	In general, the grid may not find enough agents to meet the reduction goal $M$, and may thus need to use some fallback option like a standby generator or emergency blackouts. For every amount $m$ we denote by $R_m$ the cost for the grid of using its fallback option to reduce consumption/increase production by $m$ units. E.g. if the total cost of \hma{total bids on any valid set of ?} demand response contracts exceeds $R_M$, then the power company is better off without assigning any contracts, or the fallback options can be used to fill up some gap between $\sum_{(i,j)\in S^*}\ell_j$ and $M$. 
	A fallback option can be simply added to the mechanism as a `virtual agent' that bids $R_m$ on the fixed contract $(\ell=m,F\equiv R_m)$. This is similar to the role of reserve prices in auctions. 
		The reserve costs guarantee that: (I)
			The grid finds a cheap set of solutions, whether these solutions are contracts with agents or external options; and (II) For any $(i,j)$, the reward is bounded: $r_i\leq R_{M-L}-R_{M-(L+\ell_j)}$ where $L=\sum_{j'\in S^*\setminus \{j\}}\ell_{j'}$. 
	\rmr{	and (III) for  any contract $(\ell,F)$ where $F(0) \geq R_\ell$, no agent can gain by taking the contract and being useless (reduce $X_i=0$).} \hma{Here comes my objection on ``cheapest" again. Agents bids and the reserve costs are different things and incomparable... we are comparing them anyway in the computation of allocation etc, but making this claim feels very dubious. I also don't think $r_i\leq R_{\ell_j}$ is correct without assuming that the grid can provide an unlimited number of $j$ units at this cost. Assume that there are two agents, all bidding on contract $j$. It's possible that both the two agents, and the reserve at level $j$ are selected by the optimal solution. The payments need to be computed through $R_{2j}$ (what's needed if an agent backs out), which can be larger than $2R_j$, thus the payment may exceed $R_j$.}

\section{Reliability and Expenses}
\label{sec:penalties}
The incentive analysis we presented goes through for any set of contracts. Yet, the goal of the grid  is to match demand and supply, preferably at  low total expenses, which is a-priori not the same as minimizing the sum of bids. By restricting DR-VCG to use structured constructs we can relate this goal. 

We next analyze how the sum of bids relates to \emph{reliability}, i.e. the probability that  the reduction target is met. \ver{A detailed example comparing reliability and payments across mechanisms is in Appendix~\ref{apx:example}.}{}

\subsection{Fixed penalty contracts}
Denote by $P(S,m)$ the probability that a quantity of at least $m$ is reduced under contracts $S$. Then $P^*(N,J,M) = P(S^*(N,J,M),M)$ is the \emph{reliability} of the DR-VCG mechanism. We would like to measure or bound $\ol P^*(N,J,M) = 1-P^*(N,J,M)$, which is the probability that the mechanism \emph{fails} to meet the lower bound reduction $M$ (we may omit some of the parameters).  
\vrb{OK, this is a nice result, especially as we can't do better in the most general case - although perhaps an example would help to give an intuition of how useful this bound is in practice. But, it is slightly "`isolated"' from the rest of the paper. Could also go IF really pressed with space.}

\begin{proposition}
Let $J = \{(\ell_j,f)\}_{j=1,2,\ldots}$ for some fixed $f$, then 
$\ol P^*(M) \leq \frac{1}{f}SB^*(M)$. This bound is tight.
\end{proposition}
\begin{proof} 
For any agent $i$ that is assigned a contract $j$:
The optimal investment is $c^*_{i}(j)$. 
The expected penalty is 
$$EF_i(j,c^*_i(j))=Pr_{X\sim p_i(c^*_i(j))}[X<\ell_j] \cdot f,$$
i.e., proportional to the probability it will undershoot the commitment $\ell_j$. 
The DR-VCG mechanism minimizes
{\small $$\vspace{-2mm}SB(S)=\sum_{(i,j)\in S}C^*_i(j) = \sum_{(i,j)\in S}c^*_i(j) + f\cdot \sum_{(i,j)\in S}Pr[X_i < \ell_j],$$}
that is, a combination of the total investment and the sum of individual failure probabilities. Note that  
{\small
\begin{align*}
 & \sum_{(i,j)\in S}Pr[X_i < \ell_j]  \stackrel{(a)}{\geq} Pr[\exists (i,j)\in S \text{ s.t. }X_i< \ell_j] \\
\stackrel{(b)}{\geq} & Pr[\sum_{(i,j)\in S} X_i < \sum_{(i,j)\in S} \ell_j ] 
\stackrel{(c)}{\geq}   Pr[\sum_{(i,j)\in S} X_i < M]. \vspace{-5mm}  
\end{align*}
}
Thus,
{\small
\begin{align*}
SB^*(M) &\geq f \cdot Pr[\sum_{(i,j)\in S^*(M)} X_i < M] +  \sum_{(i,j)\in S^*(M)}c^*_i(j) \\
 &= f\cdot \ol P^*(M) +  \sum_{(i,j)\in S}c^*_i(j)  \stackrel{(d)}{\geq} f \cdot \ol P^*(M). \vspace{-2mm} 
\end{align*}
}

To see why the bound is tight, observe that the inequalities in the proof are tight if (respectively): (a) 
failure events are disjoint (i.e. maximally negatively correlated);
 (b) 
agents never reduce more than $\ell_j$; 
 (c) 
 the reduction goal is met exactly ($\sum \ell_j = M$); and (d) 
 investments are $0$. 
\end{proof}

Thus a fixed penalty lets us bound the probability that the grid fails (reduction goal is not met). As we increase $f$, the (bound on) failure probability becomes smaller, at higher expense (due to higher bids). \hma{I think the claim is true, that as $f$ increases, the overall reliability is going to increase. However, the reason for this is not that the bound gets tighter... $S^\ast$ would also increase with $f$, and it's not obvious that it should increase slower than $f$ does...}\rmr{I think it's ok}
Of course, this is a worst-case bound. If, for example, some agents exceed their commitment then this would compensate for failures of others, and will increase the probability that the reduction goal $M$ is met. 

In general the grid may set a higher goal $M'=\gamma M$ than the expected surge $M$ as a safety margin. 
Another result ties this safety margin with the sum of bids.
\begin{proposition}\label{prop:F_margin}
Suppose that there is a single Fixed contract $(1,f)$. Then $M'-E[\sum_{i\in S^*(M')} X_i] \leq \frac{1}{f}SB^*(M')$.
\end{proposition}
\rmr{better as $E[\sum_{i\in S^*(M')} X_i] \geq M'-\frac{1}{f}SB^*(M')$?}
\rmr{remove proof before submission
\begin{proof}
Since all contracts over one unit of energy, we have $|S^*|=M$, and the reliability of each agent $i$ is $p_i$. Once again $B_{ij}=(1-p_i)f + c^*_i(j) \geq (1-p_i)f$.
Thus for any set $S$, 
\begin{align*}
E&[\sum_{i\in S} X_i] = \sum_i E[X_i] = \sum_i p_i = |S|-\sum_{i\in S}(1-p_i) \\
&\geq M' - \sum_{(i,j)\in S}\frac{1}{f}B_{ij} = M'-\frac{1}{f}SB(S),
\end{align*}
and this holds in particular for $S=S^*(M')$.
\end{proof}
}
Thus if the grid sets $M'$ s.t. 
$M'-\frac{1}{f}SB^*(M') \geq M$, then actual reduction is at least $M$ \emph{in expectation}. 
\rmr{it might not be possible, but at least we would know after trying}
\subsection{Cliff penalty contracts}

We saw that having a constant penalty for a violation allows us to bound the failure probability.
A Cliff penalty is more ``forgiving,'' yet it provides similar guarantees.

\begin{proposition}\label{th:cliff_M}
Suppose that the set of possible contracts $J$ is composed of Cliff penalty contracts of the form $(\ell_j,f,\alpha,\beta_j)$ for some fixed $f$ and $\alpha$ (same for all contracts), then 
$\ol P(S,\alpha \cdot M) \leq \frac{1}{f} SB(S)$.  This bound is tight.
\end{proposition}

That is, we get a guarantee on the probability that we miss the reduction goal by a factor of $\alpha$ (tightness is achieved if either $\alpha=1$ or $\beta_j=0$ for all $j$). Note that this does not require any assumption on the types of the agents. The grid can then sign contracts that sum up to $M' = \frac{M}{\alpha}$ so as to bound the probability of missing its actual goal $M$.



\rmr{removed:
\newpar{Interval contracts} To guarantee demand-supply match, the grid may want to also place an upper bound on the total reduction, i.e. prefer contracts with lower overall uncertainty. This is possible by using  contracts with penalties for excessive reduction (Interval contracts, as in Fig.~\ref{fig:contracts}(c)) and modifying the optimization goal accordingly. We leave a detailed study of Interval contracts to future work.
}

\section{DR-SCE vs. DR-VCG}
We argue that the DR-SCE mechanism can be simulated exactly using a Cliff payment contract. Formally, in DR-SCE and each agent submits a bid $b_i$, and the grid selects agents at random until $\sum b_i\geq M$. 
  After the reduction $X_i$ is realized, each selected agent gets a reward of $r_i=0$ if $X_i<b_i/2$,  $r_i=\frac{X_i}{2}$ if $X_i\in [b_i/2,3b_i/2]$, and   $r_i= 3b_i/2$ otherwise.
 
We argue that same SCE contracts used today can be offered via the DR-VCG mechanism.
For any $\ell>0$, we define a Cliff contract $j_\ell = (\ell,f_\ell = \frac{\ell}{2},\alpha = \frac13, \beta = \frac12)$. 
\begin{proposition}\label{prop:SCE_is_VCG}
For any agent $i$ of type $p_i$, submitting optimal bid $b_i$ to the DR-SCE mechanism is ex-post equivalent to being the only bidder in DR-VCG with $M\geq b_i$, $J^{SCE} \hspace{-0.2em}=\{j_\ell\}_{\ell\geq 0}$, and reserve prices $R_m=m/2$ for all $m$.  
\end{proposition}
\begin{proof}[Proof sketch]
We show that contract $b_i$ in DR-SCE is completely equivalent to contract $j_\ell$ where $\ell=3b_i/2$ (i.e. same behavior and same ex-post utility). This is by writing  the penalty function $F(j_\ell,X)$, and considering the realization of $X_i$ when: $X_i< b_i/2$, $X_i\in [b_i/2, 3b_i/2]$, and $X_i>3b_i/2$. 
Then bidding $b_i$ is optimal in DR-SCE if and only if DR-VCG assigns $j_\ell$ to $i$.  
\hma{It's obvious that the contracts are equivalent, but it took me a while to understand why having this single agent in VCG is equivalent to SCE. If we have space, how about providing some intuition like there is some profit for the agent for bidding each amount under SCE, and the agent optimizes the bid. This bid would also be chosen by VCG since it also minimizes the total summation of bids and reserve costs combined?}
\end{proof}

\medskip
Proposition~\ref{prop:SCE_is_VCG} has two important implications. First, transition from the currently used DR-SCE mechanism to DR-VCG can be gradual and backward-compatible: we can still allow bids on quantity ($b_i$) and internally convert them to the appropriate Cliff contract $j_{\ell}$ with reserve price $0.5\ell$.  \hma{I'm confused. Quantities can be converted to Cliff contracts with no problem. If people only bid on quantity, we can't really infer their cost?} \rmr{you know their cost is below the reserve price, so you can still assign them to this contract as if they bid the reserve price 0.5}
Second, it becomes obvious that DR-SCE is just a very restricted version of the more general DR-VCG, where parameter values are arbitrary and most likely suboptimal. \hma{might also be confusing. We has just said that it's equivalent to single agent VCG's so the whole DR-SCE is not equivalent to a single VCG. We can say that contracts used in DR-SCE are very restricted?} By setting the proper reduction target and reserve prices, we expect DR-VCG to outperform DR-SCE. In particular:
\begin{enumerate}[labelindent=0pt]
	\item DR-SCE makes no informed selection. With an explicit reduction goal $M$ (based on the actual surge prediction), DR-VCG \textbf{selects agents} who are more reliable. Thus we expect that $P^{VCG} >  P^{SCE}$ in most cases.
		\item DR-VCG pays rewards based on \textbf{competition}. In fact, for $J^{SCE}$ and any $M$ and $S$, $TE^{VCG}\leq TE^{SCE}$ \hma{expense notation. Also here's the same problem if we consider this a theoretical bound, as the beginning of the section 4. We can say empirically the average expense of VCG would be lower than that of SCE?}. \rmr{this claim is trivial, but it does not imply that VCG actually pays less since S may differ}
		\item DR-VCG allows agents to bid on \textbf{multiple contracts}, so they can reveal more information on their type.
		\item DR-SCE uses arbitrary  price of $0.5m$ for contracts of size $m$, whereas DR-VCG is \textbf{flexible}. In particular we may use the actual costs of generating $m$ kWh, which are highly non-linear due to the cost of adding another generator. \rmr{The average cost of producing 1 kWh with gas turbine was $\sim \$0.0332$ at 2015 (\url{http://www.eia.gov/electricity/annual/html/epa_08_04.html}), which is much less than $\$0.5$. However costs are likely to be non-linear: each gas turbine generates hundreds of MWh, while the expected surge $M$ may be just a few MWh. Thus the cost of the outside option (keeping an additional turbine running) looks like a step function, where $R_m$ is very low  until reaching production limit, and then leaps. The marginal cost could even be $\$10$ per kWh or more.  }
	\vrb{Correct, these are sometimes called "`hockey stick"' prices.}
	\end{enumerate}
\begin{exmp}\label{ex:3_bidders_SCE}
Consider the same 3 agents from Example~\ref{ex:3_bidders_VCG}. 
In the DR-SCE mechanism, all agents will submit a bid of $b_i=100$, and the grid cannot distinguish between them. 
If it selects two of them at random, it pays 
 $TE^{SCE}=(50+45+35)\cdot\frac23 = 83.33$---much more than $TE^{VCG}=25$.

If we compare failure probabilities, then $\ol P^{SCE}=\frac13 0.1 + \frac13 0.3 + \frac13 (1-0.9\cdot 0.7) = 0.224$, whereas 
 $\ol P^{VCG}=0.1$, which is again an improvement over SCE.  
\end{exmp}
On the other hand, if agents have to invest high costs $c^*_i$ then they might not participate in DR-SCE at all, as their reward is bounded by \$0.5. 
\rmr{
 Suppose now that each agent has to invest $c^*_i=\$60$, where $p_i$ remains the same. Then the rewards under DR-VCG would increase to $r_1=r_2=75$, or $150$ (minus $5$) in total. This may seem high, but note that under DR-SCE, the agents will not bid at all! 
The example demonstrates
}
Thus DR-SCE pays \emph{too much} to agents with low $c^*_i$, and \emph{too little} to agents with high $c^*_i$. \hma{hmm but on average these agents are more expensive than 0.5 per unit... need to incentivize why the mechanism should pay more to more expensive agents for the turbine scheme? Maybe add a section 5.2 non-linear reserve costs or something and move this there and create a simple example? (if time and space, or course, and it seems we don't have any...)}

The parameters of the Cliff contract $j_\ell$ are also arbitrary (e.g. why $\alpha=\frac13$?), however there is no obvious way to set them a-priori (see Discussion). 
\rmr{I think we can prove that as long as not all agents are needed then the power company saves a substantial amount: once since it only selects the required agents, and second time because they will be paid much less than the reserve price. In the worst case all agents are selected and then the cost is exactly as in the current mechanism.
} \hma{will do and should be easy} \rmr{I actually have a counter example for TE(VCG)<TE(SCE)  when both use selection: suppose that agents 1+2 can cut 100 w.p. 1, and agent 3 can cut 200 w.p. 0.5. Thus the SCE bids are 100,100, and 200 (kWh) and the VCG bids are 0,0,and 50 (\$). Then VCG selects $\{1,2\}$ and pays 50 to each (100 in total). SCE select $\{1,2\}$ or $\{3\}$ at random, and pays either 100 in the first case, or just 50 in the latter case, and 75 in expectation. Note however that this is because SCE is much less reliable. Note that this does not contradict the statement about  $TE^{VCG}(S)\leq TE^{SCE}(S)$ above, since the mechanisms select different sets $S$. }

\subsection{Simulations}  

\paragraph{Settings} Each agent $i$ has $T$ ``effort levels,'' where each level is a triple $(c_{it},q_{it},p_{i})$, meaning that with investment $c_{it}$ agent $i$ can reduce $q_{it}$ kWh.  The reduction succeeds w.p. $p_{i}$, and w.p. $1-p_i$ reduction is 0 due to an unexpected event. The expected demand surge is $M$, and the grid uses a safety margin $\gamma \geq 1$. 
For each mechanism we denote by $TE=TE(S^*(N,J,\gamma M))$ the total expense
, and by $P =P(S^*(N,J,\gamma M),M)$ the reliability, i.e., the probability that $\sum X_i >M$ when the mechanism collects contracts for $\gamma M$. We run both DR-SCE and DR-VCG mechanisms, and measure $TE^{SCE},TE^{VCG},P^{SCE}$, and $P^{VCG}$. 
 
To set up a realistic scenario of a typical demand response event, we used \cite{DBP_eval} that summarize previous DR programs. We fix the expected demand surge to $M=10$MWh. In each economy we sample $n$ agents  i.i.d., where each agent has $T\in \{1,3,5\}$ effort levels. For each agent $i\in N$ and effort level $t\leq T$: the capacity (in kWh) is $q_{it}\sim Zipf(1,500) \cdot 10$; individual reliability is $p_{i}\sim U[0.7,1]$\rmr{we can also use a skewed distribution here, with many reliable agents and few ``rotten apples''. This will only improve VCG}; and agents' investment costs (in \$) are $c_{it}\sim U[0.2,1]$, multiplied by $q_i$.\footnote{This roughly mimics the aggregate statistics in the data, where agents' bids are highly skewed, with a minimum of $10$ kWh up to several MWh, overall reliability is $\sim 0.85$, and participation is low (about 100 bidders out of 1000 registered users). We also tried different distributions and got similar results.}
 Note that only agents with $\max_t \frac{c_{it}}{q_{it}}\leq 0.5$ will submit bids in DR-SCE. We generated populations of 3 sizes: $n=100$ (small), $n=200$ (medium) and $n=400$ (large), and for each population varied the safety margin between $\gamma\in[1,2]$. \hma{I think it's okay to mention that in the doc there are 100 agents that had submitted bids, and when $n=400$ roughly this number of people actually bid in SCE given the above setup? Or we can move the footnote to the end of this paragraph, adding this piece of information in there?} 

We run simulations that demonstrate the four advantages of the DR-VCG mechanism mentioned above. Every datapoint in our simulations is an average over 100 instances. 
\newpar{Selection and Competition}
In our first simulation, agents each have a single effort level ($T=1$). 
We use the set of ``SCE-like'' contracts $J^{SCE}=\{j_\ell\}_{\ell=10,20,\ldots}$, and set linear reserve prices $R_m=0.5m$ for all $m\geq 0$. Thus for a single bidder, DR-SCE and DR-VCG are equivalent by Prop.~\ref{prop:SCE_is_VCG}. 

Fig.~\ref{fig:large_econ} shows the expense-vs.-reliability frontier under both mechanisms.  Larger safety margin $\gamma$ results in more recruited agents and higher reliability, but also higher costs\hma{total expenses?}. We can see that in both populations DR-VCG dominates DR-SCE by guaranteeing \emph{any} reliability level at a much lower cost. \hma{total expenses?}  We found that even if we pay the reserve prices to all selected agents, DR-VCG does somewhat better than DR-SCE, meaning that it does indeed select better agents.  
\rmr{With fixed costs the reliability of DR-VCG is even better, as bids are better proxies for agents' individual reliability.}

Fig.~\ref{fig:large_econ} also shows that the advantage of DR-VCG becomes larger in large populations (or when the expected surge is small), as competition drives prices down. In contrast, in small populations DR-VCG and DR-SCE are the same, as both exhaust all agents with low investment costs. Our next simulations show how the other two advantages of DR-VCG overcome this problem.

%
%
%

\begin{figure}[t!]
\centering     
\subfloat[\small{$n=400$}]{\label{fig:fig1a_vcg_sce_n400}
	\includegraphics[scale=0.69]{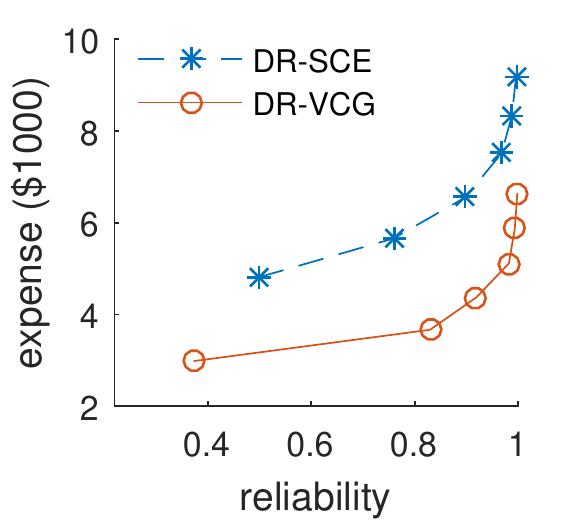}
}
\hspace{1em}
\subfloat[\small{$n=200$}]{\label{fig:fig1b_vcg_sce_n200}
  \includegraphics[scale=0.69]{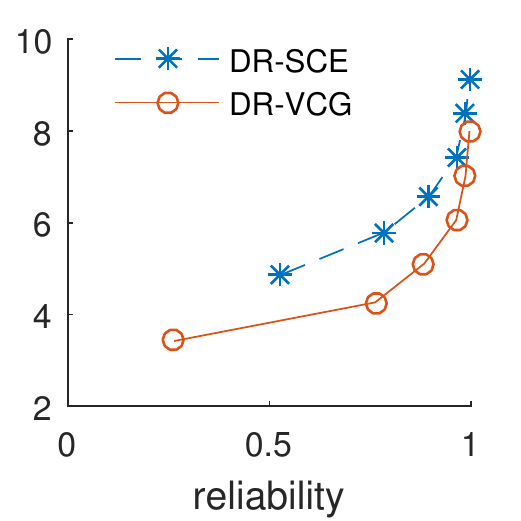}

}
\vspace{-0.5em}
\caption{Reliability vs. expense for large and medium populations and single effort level. 
\label{fig:large_econ}}\vspace{-0.0em}
\end{figure}

\newpar{Multiple levels} 
Fig.~\ref{fig:multilevel} shows the \hma{add ``expense-"}reliability frontier for a medium population ($n=200$) of agents with multiple effort levels. We can see from  Figs.~\ref{fig:large_econ} and \ref{fig:multilevel} that DR-VCG performance becomes better as population gets larger and/or agents' types are more complex, wheres the performance of DR-SCE remains almost the same. Intuitively, selecting from $n$ agents each submitting $T$ independent bids is similar to selecting from $Tn$ agents, i.e. there is more competition. 
\rmr{removed this since actually not true with fixed $p_i$ levels. I suspect VCG takes the smaller levels $q_it$ with lower marginal cost, however I don't really know --- We can also think of DR-VCG as incentivizing agents to invest more effort, and }

\begin{figure}[t!]
\centering     
\subfloat[\small{$T=3$}]{\label{fig:fig2a_vcg_sce_T3}
  \includegraphics[scale=0.69]{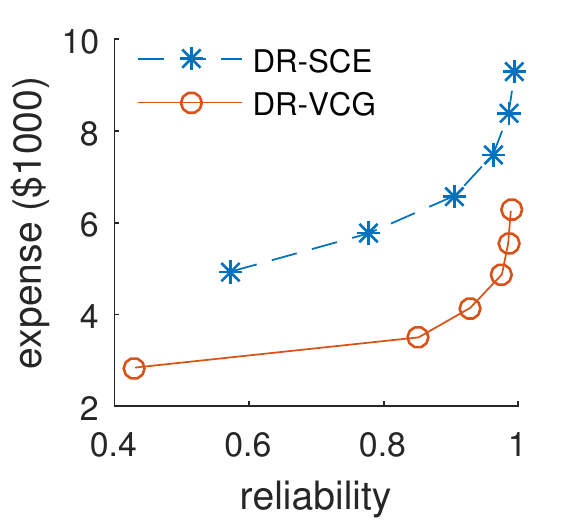}
}
\hspace{1em}
\subfloat[\small{$T=5$}]{\label{fig:fig2b_vcg_sce_T5}
  \includegraphics[scale=0.69]{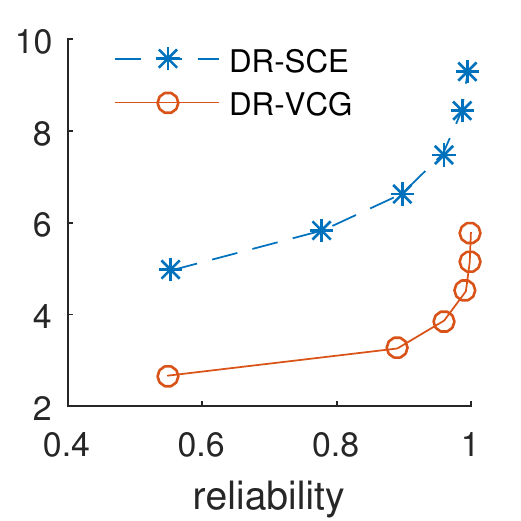}
}
\vspace{-0.5em}
\caption{Reliability vs. expense with for medium population  $n=200$, with multiple effort levels. \hma{simplified the captions a little bit... unfortunately the gap didn't increase significantly from 3 to 5, but we can mention that this is big decrease in costs in comparison to 2(b)?}
\label{fig:multilevel}}\vspace{-0.0em}
\end{figure}

\newpar{Flexible reserve prices}
The SCE contracts with their linear reserve price do not reflect correctly the outside options available to the grid. In reality, the grid cannot generate additional power at small quantities to fill the gaps between agents' bids and the reduction goal. Failure to reach the reduction goal $\gamma M$ means that the grid cannot rely on the current DR contracts, and must increase supply by operating another generator at a large cost.  The operating cost with modern gas turbines is $\$0.04-0.1$ per kWh, 
but each turbine generates at least $100$ MWh. We thus set the reserve prices to $R_m=4000+0.1m$, which creates a dichotomy between ``success'' (where the DR mechanism collected enough contracts to forgo the additional turbine) and ``failure.'' Increasing the reserve prices also requires higher penalties. Otherwise agents may bid for contracts they do not plan to keep, with reward higher than the maximal penalty. We did not optimize the contract (see Discussion) and instead just set the penalties to  $f_\ell = \ell$ (double from $j_\ell$). 
									\rmr{Reviewers may ask why we took the lower bound for the fixed cost and the upper bound for the additional cost. We can say that (a) this way remains closest to the previous linear reserve prices, (b) the simulations never gets to reserve prices anyway so different parameters would not change anything}

Fig.~\ref{fig:turbine} shows how flexible prices benefit DR-VCG in \emph{small populations}. As the target capacity $\gamma M$ increases, this requires high reliability using a small population, which DR-SCE very often fails to achieve. This is because it may not find enough reliable agents willing to bid for a payment of $\$0.5$, and it must use the extra turbine for a high cost. In contrast, DR-VCG can increase the reward to agents, thereby attracting also bidders with investment costs higher than $\$0.5$. 

\begin{figure}[t!]
\centering     
\subfloat[\small{frac. of failure}]{\label{fig:fig3a_vcg_sce_n100}
\includegraphics[scale=0.69]{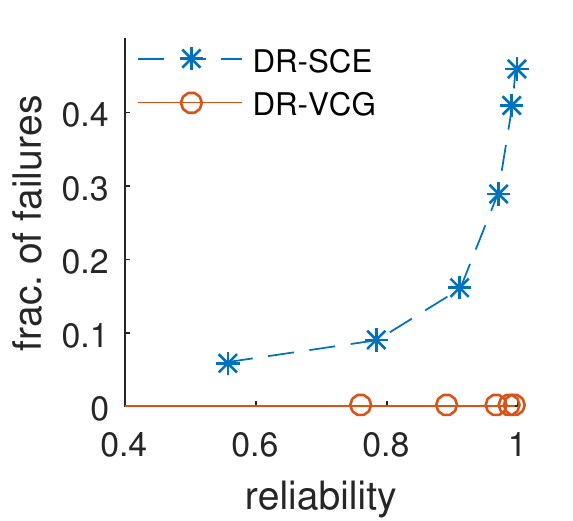}
}
\subfloat[\small{expense$\sim$reliability}]{\label{fig:fig3b_vcg_sce_n100}
	\includegraphics[scale=0.69]{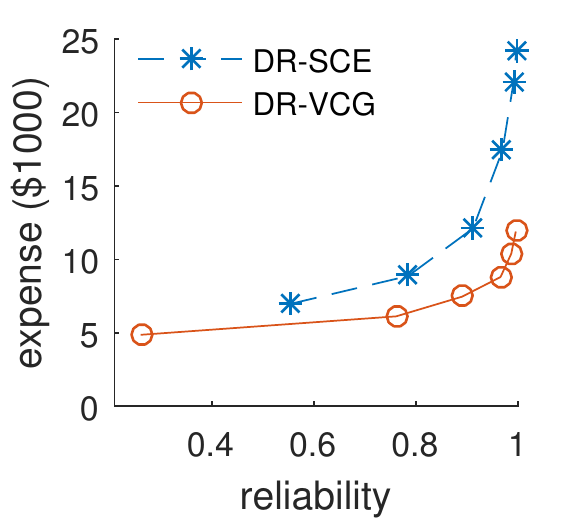}
}
\vspace{-0.5em}
\caption{Outcomes under SCE and VCG mechanism for small population ($n=100$) and flexible reserve prices. 
\label{fig:turbine}}\vspace{-0.5em}
\end{figure}

\vspace{-0mm}
\section{Discussion}
\vspace{-0mm}
We suggested in this paper the DR-VCG mechanism for demand-response contracts that is based on individual ``soft'' commitments and flexible penalty schemes.
While the details of the contracts and the analysis were specific to demand response programs, the general idea of offering flexible penalty contracts to multiple agents may be useful in other domains that require joint effort under uncertainty~\cite{porter2008fault}.   

 We considered three natural parametric classes of penalties that allow for efficient computation of VCG prices, and showed how they generalize the currently deployed SCE contract. Power companies can adopt the new DR-VCG mechanism with the SCE contracts for painless migration at first. Then, they can gradually add more contracts and optimize their parameters based on distributional assumptions, data on consumption profiles, trial-and-error, and so on. Another benefit of DR-VCG is that the grid can focus on optimizing the set of contracts without worrying about agents' strategic behavior, whereas agents can focus on accurately estimating their costs for each contract. Based on initial simulations, it seems that penalties should be higher than $\$0.5$ per kWh, and perhaps superlinear in the size of the contract (as reliability of large consumers is more important).
Finding  optimal parameters under various assumptions is a major topic for future research, as well as a better understanding of the connections between reliability, penalties, and expenses.


\ver{}{
\clearpage
}
\bibliographystyle{named}
\bibliography{ijcai17-dr-refs}

\begin{thebibliography}{}

\bibitem[\protect\citeauthoryear{Akasiadis and
  Chalkiadakis}{2013}]{AkasiadisC13}
Charilaos Akasiadis and Georgios Chalkiadakis.
\newblock Agent cooperatives for effective power consumption shifting.
\newblock In {\em Proceedings of the Twenty-Seventh {AAAI} Conference on
  Artificial Intelligence}, 2013.

\bibitem[\protect\citeauthoryear{Chapman and Verbic}{2017}]{chapman_verbic}
A.~C. Chapman and G.~Verbic.
\newblock An iterative on-line auction mechanism for aggregated demand-side
  participation.
\newblock {\em IEEE Transactions on Smart Grid}, 8(1):158--168, Jan 2017.

\bibitem[\protect\citeauthoryear{Clarke}{1971}]{clarke1971}
E.~H. Clarke.
\newblock Multipart pricing of public goods.
\newblock {\em Public choice}, 11(1):17--33, 1971.

\bibitem[\protect\citeauthoryear{Hansen \bgroup \em et al.\egroup
  }{2014}]{DBP_report}
Daniel~G. Hansen, Steven~D. Braithwait, David~A. Armstrong, and Marlies~C.
  Hilbrink.
\newblock {2013 Load Impact Evaluation of California Statewide Demand Bidding
  Programs (DBP) for Non-Residential Customers: Ex Post and Ex Ante Report }.
\newblock
  \url{http://www.calmac.org/publications/PY13_DBP_Ex_Ante_Report_20140401_PUBLIC_Revised.pdf},
  2014.

\bibitem[\protect\citeauthoryear{Haring \bgroup \em et al.\egroup
  }{2016}]{haring_al}
T.~W. Haring, J.~L. Mathieu, and G.~Andersson.
\newblock Comparing centralized and decentralized contract design enabling
  direct load control for reserves.
\newblock {\em IEEE Transactions on Power Systems}, 31(3):2044--2054, May 2016.

\bibitem[\protect\citeauthoryear{Karp}{1972}]{karp1972reducibility}
Richard~M Karp.
\newblock Reducibility among combinatorial problems.
\newblock In {\em Complexity of computer computations}, pages 85--103.
  Springer, 1972.

\bibitem[\protect\citeauthoryear{Li \bgroup \em et al.\egroup }{2015}]{NaLi_al}
Na~Li, Lijun Chen, and Munther Dahleh.
\newblock Demand response using linear supply function bidding.
\newblock {\em IEEE Transactions on Smart Grid}, 6(4):1827--1838, 2015.

\bibitem[\protect\citeauthoryear{Ma \bgroup \em et al.\egroup
  }{2016}]{Ma_ijcai16}
Hongyao Ma, Valentin Robu, Na~Li, and David~C. Parkes.
\newblock {Incentivizing Reliability in Demand-Side Response}.
\newblock In {\em Proceedings of the 25th International Joint Conference on
  Artificial Intelligence (IJCAI'16)}, 2016.

\bibitem[\protect\citeauthoryear{Ma \bgroup \em et al.\egroup
  }{2017}]{ma_AAMAS17}
Hongyao Ma, David~C. Parkes, and Valentin Robu.
\newblock {Generalizing Demand Response Through Reward Bidding}.
\newblock In {\em Proceedings of the 16th International Conference on
  Autonomous Agents and Multiagent Systems, (AAMAS'17)}, 2017.

\bibitem[\protect\citeauthoryear{Meir \bgroup \em et al.\egroup }{2017}]{MMR17}
Reshef Meir, Hongyao Ma, and Valentin Robu.
\newblock Contract design for energy demand response.
\newblock In {\em Proceedings of the 26th International Joint Conference on
  Artificial Intelligence (IJCAI'17)}, 2017.

\bibitem[\protect\citeauthoryear{Nekouei \bgroup \em et al.\egroup
  }{2015}]{nekouei_al}
E.~Nekouei, T.~Alpcan, and D.~Chattopadhyay.
\newblock Game-theoretic frameworks for demand response in electricity markets.
\newblock {\em IEEE Transactions on Smart Grid}, 6(2):748--758, March 2015.

\bibitem[\protect\citeauthoryear{Nisan}{2007}]{nisan2007introduction}
Noam Nisan.
\newblock Introduction to mechanism design (for computer scientists).
\newblock In Nisan et~al., editor, {\em Algorithmic game theory}. Cambridge
  University Press, 2007.

\bibitem[\protect\citeauthoryear{Patterson \bgroup \em et al.\egroup
  }{2014}]{DBP_eval}
Olivia Patterson, Mary Sutter, and Alan Elliott.
\newblock {2012-2013 PG\&E and SCE Demand Bidding Program Process Evaluation}.
\newblock
  \url{http://www.calmac.org/publications/California_2012-2013_Demand_Bidding_Program_Process_Evaluation_FINAL.pdf},
  2014.

\bibitem[\protect\citeauthoryear{Porter \bgroup \em et al.\egroup
  }{2008}]{porter2008fault}
Ryan Porter, Amir Ronen, Yoav Shoham, and Moshe Tennenholtz.
\newblock Fault tolerant mechanism design.
\newblock {\em Artificial Intelligence}, 172(15):1783--1799, 2008.

\bibitem[\protect\citeauthoryear{Rose \bgroup \em et al.\egroup
  }{2012}]{rose_al}
Harry~Thomas Rose, Alex Rogers, and Enrico~H. Gerding.
\newblock A scoring rule-based mechanism for aggregate demand prediction in the
  smart grid.
\newblock In {\em International Conference on Autonomous Agents and Multiagent
  Systems, {AAMAS} 2012}, pages 661--668, 2012.

\bibitem[\protect\citeauthoryear{Samadi \bgroup \em et al.\egroup
  }{2012}]{samadi2012advanced}
Pedram Samadi, Hamed Mohsenian-Rad, Robert Schober, and Vincent~WS Wong.
\newblock Advanced demand side management for the future smart grid using
  mechanism design.
\newblock {\em IEEE Transactions on Smart Grid}, 3(3):1170--1180, 2012.

\bibitem[\protect\citeauthoryear{Su \bgroup \em et al.\egroup }{2014}]{Su_al}
Wencong Su, Jianhui Wang, and Jaehyung Roh.
\newblock Stochastic energy scheduling in microgrids with intermittent
  renewable energy resources.
\newblock {\em IEEE Transactions on Smart Grid}, 5(4):1876--1883, 2014.

\bibitem[\protect\citeauthoryear{Zhang \bgroup \em et al.\egroup
  }{2015}]{Zhang_al}
Baosen Zhang, R.~Johari, and R.~Rajagopal.
\newblock Competition and coalition formation of renewable power producers.
\newblock {\em IEEE Transactions on Power Systems}, 30(3):1624--1632, 2015.

\end{thebibliography}

\clearpage
\onecolumn
\appendix

\section{Variants of DR mechanisms}
\label{apx:variants}
\paragraph{Agent-side variant}
\rmr{Omit from full version?}
This variant is similar to the auction variant, only instead of bidding on some fixed set of contracts $J$, the grid publishes a parametric penalty scheme $F$. Each agent submits a finite number of bids, where each bid contains a contract $j$ (i.e. its parameters) and a cost $c_{ij}$. For example, the grid publishes the fixed scheme $f_\ell = \ell/2$ as above. Agent~1 submits two bids, $J_1=\{(5,8),(7,10)\}$ (i.e.  the agent asks \$8 to reduce consumption by $5$ units, or \$10 to reduce 7 units), agent~2 submits the bids $J_2=\{(5,6), (10,12),(100,20)\}$, and so on. The grid then runs the auction mechanism with $J=\bigcup_{i\in N}J_i$. 

This mechanism leaves the decision of what contracts to bid on for the agent, thereby allowing them to submit fewer bids on goals that are convenient to them.

\paragraph{Direct revelation variant}
The grid publishes a parametric penalty scheme $F$, as in the agent-side variant. Then each agent reports her entire cost function which maps every possible contract  to a cost under $F$. For example, the grid may publish the fixed penalty scheme $f_\ell = \ell/2$. Then agents will submit their reported cost functions $C_i(\ell)$ in some concise form. 
The mechanism then chooses for each agent $i$ a contract $\ell^*_i$ from the (infinite) set $\{(\ell,C_i(\ell))\}_{\ell\geq 0}$ (and the null contract $(0,0)$), such that  $S^*(N)=\{(i,\ell^*_i)\}_{i\in N}$ minimizes the total cost $\sum_{i\in N}C_i(\ell^*_i)$ among all valid contracts.

This version of the mechanism is the most demanding one: both for the agents who should come up with a function describing their cost for any possible contract; and for the grid that has the burden of optimizing over an infinite set of contracts. On the other hand, since the set of contracts in this version is the largest one, the outcome is better in terms of social cost.

\section{Proofs}
\label{apx:proofs}
\subsection{Computational Complexity}
\begin{proposition}Checking whether $SB^*(N,J)\leq Z$ for some input $N,J,Z$ is NP-hard, even when $J$ includes only fixed contracts.
\end{proposition}
\begin{proof}
Proof is by a  reduction from the Knapsack problem~\cite{karp1972reducibility}.
Given a Knapsack instance $\{(v_i,w_i)\}_{i\in N}$ (volume, worth), we define for each item $i\in N$  a fixed contract $j_i=(w_i,1)$, and an agent $i$  such that  $B_{i, j_i}=v_i$ and $B_{i,j}=\infty$ for any $j\neq j_i$. Then $SB^*(N,J)\leq Z$ if and only if there is a set of items of total worth $M$ that fit in a sack of size $Z$.
\end{proof}

\begin{rtheorem}{th:complex} For any sets of agents $N$ and Cliff contracts $J$, both of $S^*(N,J)$ and $SB^*(N,J)$  (and thus also VCG prices) can be computed in time polynomial in $n,k,M$. 
\end{rtheorem}
 		
		\begin{proof}
		We provide a dynamic program that decides in poly time whether $SB^*(N,J)\geq Z$. We can then solve the optimization problem by doing binary search on the value of $Z$. 
		
		Fix an arbitrary order over agents in $N$, and initialize tables $S$ and $SB$, each of size $(n+1)\times M$. 
		
		\begin{enumerate}
			\item The cell $S(0,m)$ is initialized to $\emptyset$ for $1 \leq m \leq M$.
		\item The cell $SB(0,m)$ is initialized to zero for $m=0$ and to infinity for $1 \leq m \leq M$.
		\item The cells $S(k,m), SB(k,m)$ will contain the optimal subset  $S\subset\{1,\ldots,k\}\times J$  s.t. $\sum_{(i,j)\in S}\ell_{ij} = m$, and the sum of bids in of $S(k,m)$, respectively. 
		\end{enumerate}
		
		Given $S(k-1,m)$ for all $m\leq M$, we can compute $S(k,m)$ for all $m\leq M$, by considering the best option to meet $m$ without agent $k$ ($SB(k-1,m)$), and all possible contracts of agent $k$, i.e. $S(k-1,m-\ell_j)\cup \{(k,j)\}$ for each $j\in J$. We take the best  solution from all $|J|+1$ options.
		
		More formally, let
		$$j^* = \argmin_{j\in J} SB(k-1,m-\ell_j) + B_{k,j},~~~~ s^* = \min_{j\in J} SB(k-1,m-\ell_j) + B_{k,j}.$$
		If $s^*\geq SB(k-1,m)$ then set $S(k,m) = S(k-1,m)$ and $SB(k,m) = SB(k-1,m)$. Otherwise, set $S(k,m) = S(k-1,m-\ell_{j^*}) \cup \{(k,j^*)\}$ and $SB(k,m) = SB(k-1,m-\ell_{j^*}) + B_{k,j^*}$.
		
		Finally, We consider all feasible solutions ($S(n,m)$ for $m\geq M$) and select the one with the minimal cost $SB(n,m)$ among them. 	
		\end{proof}

\subsection{Incentives}
\begin{rtheorem}{th:auction}
Consider the auction variant of the DR-VCG mechanism.
\begin{enumerate}
	\item For every contract $j\in J$, it is a dominant strategy for agent $i$ to bid   $C^*_{i}({j})$;
	\item If contract $(i,j)$ is selected, it is a dominant strategy for $i$ to invest $c^*_{i}(j)$ in preparation; 
	\item The mechanism is IR.	
\end{enumerate}
 
\end{rtheorem} 
 \begin{proof}
First, suppose  contract $(i,j)$ was selected. Thus $i$'s expected utility for investing $c$ is $v_i(c,j)=r_i-c-F_i({j},c)$.  By definition, this is maximized by investing $c^*_i({j})$. Thus we can assume that agents indeed invest the effort on which they base their overall true cost $C^*_i(j)$.


Let $i^j$ be an agent that attaches $0$ cost to contract $j$, and infinite cost to all other contracts.
It is easy to show that the mechanism can be interpreted as follows: \rmr{perhaps cite some similar paper?}
\begin{itemize}
	\item ``cost-independent prices": for each agent $i$ and each contract $j$, the mechanism offers a payment:
\begin{align*}
	t_i^{j} = SB^*( N_{-i}) - SB^*_{-i}(N_{-i}\cup\{i^j\}) 
\end{align*}
where $SB^*_{-i}(N_{-i}\cup\{i^j\})$ is the social cost on the rest of the agents, when agent $i$ (or $i^j$) gets contract $j$.
	\item ``agent-maximizing selection": each agent selects the utility-maximizing contract from the set of offered contracts, unless all contracts yield negative utility.
\end{itemize}
To see why, suppose that  $i$ strictly prefers contract $j'$ over contract $j$. Then 
$$t_i^j-C^*_i({j}) < t_i^{j'}-C^*_i({j'})$$
then  \hma{Some notation problem in the following inequality since the direction should be reversed --- I'm too brain dead at this moment to check...}\rmr{now ok?}
\begin{align*}
SB^*(N) &= SB^*_{-i}(N_{-i}\cup\{i^j\}) + C^*_i(j) = (SB^*(N_{-i})-t_i^j) + C^*_i(j) =SB^*(N_{-i})-(t_i^j - C^*_i(j)) \\
& > SB^*(N_{-i})-(t_i^{j'} - C^*_i(j')) = SB^*_{-i}(N_{-i}\cup\{i^{j'}\}) + C^*_i(j') = SB(S'),
%
%
\end{align*}
That is, the mechanism would prefer an alternative assignment $S'$ where $i$ is assigned to contract ${j'}$ rather than $j$. 
Since $t_i^j$ are independent of  agent $i$'s reported costs, it is a dominant strategy to report the true expected cost $C^*_i(j)$. It is left to show that the mechanism is IR, i.e. that agent $i$ selects a contract only if this guarantees a nonnegative utility.  
\begin{itemize}
		\item If $(i,j) \in S^*(N)$, $SB^*( N_{-i}) \geq SB^*(N)$ 
		must hold. Therefore, 
		$$t_i^j - C^*_i(j) = SB^*( N_{-i}) - SB^*_{-i}(N_{-i} \cup \{i^j\}) 
		- C^*_i(j) = SB^*( N_{-i}) - SB^*(N) \geq 0.$$
		\item If $i \notin S^*(N)$, then $SB^*( N_{-i}) =SB^*_{-i}(N)$, thus for all $j\in J$, $t_i^j - C^*_i(j) = 
			0- C^*_i(j) \leq 0$.
	\end{itemize}
\end{proof}

The theorem has the following corollaries on the other two variants we presented:
\begin{enumerate}
	\item The direct revelation mechanism is truthful, as it is a dominant strategy for an agent to report her full cost type. This is simply by considering $J$ as the set of all contracts. \hma{Not necessarily. I think we can still fix the same set of intervals but just do the math for the agents with their reported types}
	\item For the agent-side variant, we do not specify how the agent selects on which contracts to bid: indeed, there may not be a dominant strategy for an agent to select contracts. Yet on any set of reported contracts, the agent will bid her true costs, and the mechanism will select the optimal subset from all reported contracts. To see why, just suppose that the agent first suggests the set of contracts to the grid, and then the grid includes them in the published set $J$.
	\hma{I think it's a bit problematic here. The agent needs to consider all possible intervals, do some math and figure out which one is the most profitable to report. This will depend on the report of the rest of the agents so it's not clear what the strategy of the agents would be like.} \rmr{better now?} 
\end{enumerate}

\subsection{Reliability Guarantees}
\begin{rproposition}{th:cliff_M}
Suppose that the set of possible contracts $J$ is composed of cliff penalty contracts of the form $(\ell_j,f,\alpha,\beta^j)$ for some fixed $f$ and $\alpha$ (same for all contracts), then 
$\ol P(S,\alpha \cdot M) \leq \frac{1}{f} SB(S)$.  This bound is tight.
\end{rproposition}
\begin{proof}
\begin{align*}
SB(S)&=\sum_{(i,j)\in S}C^*_i(j) = \sum_{(i,j)\in S}c^*_i(j) +  \sum_{(i,j)\in S}E[F_i(j)] \\
& \geq  \sum_{(i,j)\in S}E[F_i(j)] \geq  \sum_{(i,j)\in S}f_{\ell_j}\cdot  Pr(X_i < \alpha^j \ell_j) \\
\end{align*}
For fixed $f_\ell = f$ and $\alpha^j=\alpha$, 
\begin{align*}
& = f \cdot  \sum_{(i,j)\in S}\cdot  Pr(X_i < \alpha \ell_j) \geq f\cdot  Pr(\sum_{(i,j)\in S} X_i < \alpha\sum_{(i,j)\in S}  \ell_j ) \\
&\geq  f\cdot Pr(\sum_{(i,j)\in S} X_i < \alpha M), 
\end{align*}
\end{proof}

\subsection{DR-SCE implemented by DR-VCG} 

For any $\ell>0$, we define a cliff contract $j_\ell = (\ell,f_\ell = \frac{\ell}{2},\alpha = \frac13, \beta = \frac12)$.

\begin{rproposition}{prop:SCE_is_VCG}
For any agent $i$ of type $p_i$, submitting optimal bid $b_i$ to the DR-SCE mechanism is ex-post equivalent to being the only bidder in DR-VCG with $M\geq b_i$, $J^{SCE} \hspace{-0.2em}=\{j_\ell\}_{\ell\geq 0}$, and reserve prices $R_m=m/2$ for all $m$.  
\end{rproposition}
\begin{proof}
The reward to the agent for contract $j_\ell$ is determined by the reserve cost, thus  $r_i=R_\ell = \ell/2$.
 The penalty function  $\tilde F$ has the following form:
$\tilde F(j_\ell,X)=\left\{
\begin{array}{ll}
\ell/2, &X<\ell/3 \\
(\ell-X)/2,  & \ell/3  \leq X < \ell \\ 
0 ,& \ell \leq X
\end{array}
\right.
 $
For any bid $b$, we identify a corresponding contract $j=j(b)$, such that $\ell_j=\frac{3b}{2}$ (the argument $b$ is omitted when clear from context).
We will show that for any realization of $X_i$, the utility of $i$ is the same under the DR-SCE (with bid $b_i$) and under the DR-VCG mechanism with the truthful bid $C^*_{i}(j)$ on contract $j=j(b)$.  Denote by $r^{SCE}_i(b,X)$ the realized reward to an agent in the SCE mechanism that bids $b$ and reduces $X$. 
Note that there are 3 cases: 
\begin{enumerate}
	\item If $X_i<b_i/2$ then $r^{SCE}_i(b_i,X_i)=0$. In this case it also holds that $X_i<\frac{3b_i}{3\cdot 2} = \frac{\ell_j}{3}$, and thus   the penalty is $F(j(b_i),X_i)=\ell_j/2$. Thus the reward minus penalty in DR-VCG is $r_i-F(j(b_i),X_i) = \ell_j/2-\ell_j/2 = 0= r^{SCE}_i(b_i,X_i)$.
	\item If $b_i/2  \leq X_i < 3b_i/2$, then $r^{SCE}_i(b_i,X_i)=X_i/2$. In this case $\ell_j/3  \leq X_i < \ell_j$ so the total payment is $r_i - \tilde F(j,X_i) = \ell_j/2 - (\ell_j-X_i)/2 = X_i/2 = r^{SCE}_i(b_i,X_i)$. 
	\item If $3b_i/2\leq X_i$, then $r^{SCE}_i(b_i,X_i)=3b_i/4$. In this case $\ell_j \leq X_i$ so the total payment is $r_i- F(j,X_i) = \ell_j/2 - 0 = 3b_i/4 = r^{SCE}_i(b_i,X_i)$.
\end{enumerate}

If there is a preparation cost, it is the same cost under both mechanisms, thus in either case the outcome is completely equivalent.

It is left to show that under both mechanisms agent $i$ ends up with the same contract.
In the DR-SCE mechanism, the bidder selects $b^*_i$ which maximizes  \hma{here, $b$ becomes $j$ after the first equality, which might be a bit confusing. How about just change the $b$'s to be $j$'s? Or we can write $b_i^\ast = \arg\max_{b \geq 0} \{ E[r^{SCE}_i(b,X)]-c^*_{ib} \} = \arg\max_{j \geq 0} \{  E[r_i-F(j,X)]-c^*_{i}(j) \}$? }
$$E[r^{SCE}_i(b,X)]-c^*_{ib}=E[r_i-F(j,X)]-c^*_{i}(j) = \ell_j/2-(E[F(j,X)]+c^*_{i}(j)).$$
 DR-VCG, on the other hand, minimizes the sum of bids that reach $M$, where the gap between the unique selected contract $\ell_j$ and $M$ is filled by the reserve.  I.e. $S^*=\{(i,j^*)\}$, where
\begin{align*}
j^*&=\argmin_j\{C^*_{i}(j)+R_{M-\ell_j}\}= \argmin_j\{E[F(j,X)]+c^*_{i}(j)+R_{M-\ell_j}\}\\
&= \argmin_j\{E[F(j,X)]+c^*_{i}(j)+\frac{M-\ell_j}{2}\}=\argmin_j\{E[F(j,X)]+c^*_{i}(j)-\frac{\ell_j}{2}\},
\end{align*}
which is the same as $\argmax_j\{\ell_j/2-(E[F(j,X)]+c^*_{i}(j))\}$. Thus $j^*=j(b^*_i)$. 
In other words,  the contract $j^*$ assigned to $i$ in DR-VCG is exactly the one corresponding to bid $b^*_i$ in the DR-SCE mechanism.
%
\end{proof}

As we showed that every SCE contract has an equivalent contract that can be oferred by DR-VCG, we denote this set of contracts by $J^{SCE}$. Consider some assignment $S$ of contracts to agents. 
\begin{proposition} For any subset of contracts $S\subseteq N \times J^{SCE}$, with $R_m=m/2$ for all $m$. Then
$TE^{VCG}(S)\leq TE^{SCE}(S)$.
\end{proposition}
\begin{proof}
Note that an agent $i$ assigned to contract $j$ will invest the same effort $c^*_i(j)$ under both mechanisms, and hence will have the same reduction in expectation (and even ex-post). Denote by $EF_i=E_{X_i\sim p_i(c^*_i(j))}[F(j,X_i)]$ the expected penalty of agent $i$ in DR-VCG.

\rmr{we need to extract this as a lemma and use for both propositions}
Each agent $i$ assigned to contract $j_\ell$ earns $f_\ell-EF_i$ in DR-SCE.
The same agent in DR-VCG earns $r_i - EF_i \leq \ell/2-EF_i = f_\ell-EF_i$.

As the total expense in each mechanism is the sum of payments, DR-SCE pays weakly more.
\end{proof}
We highlight that the proposition \emph{does not entail} that DR-VCG always pays less. This is since for the same population, DR-SCE and DR-VCG may assign different sets of contracts.
 
\section{A Numeric Example with continuous intervals} \label{apx:example}\hma{I wasn't able to walk through the numbers...}

To demonstrate both our mechanism and its advantage over the current system, we consider a simple numeric example with only two agents (and no preparation costs).  We only present some details here, and the rest can be found in Appendix~\ref{sec:example}.

Suppose there are two agents with the following capacity distributions at cost $0$:  $X_1\sim U[100,200]$ and $X_2\sim U[50,250]$. Thus both agents have the same capacity in expectation but agent~1 is more reliable.  
The goal of the grid is to cut $M=150$ KW. 

Let us consider first agents' bids under  the DR-SCE mechanism (recall that the price is $\$0.5$ per KW). The dominant strategy for agent~1 is to bid anything s.t. $[100,200]\subseteq [b_1/2,3b_1/2]$, thus say $b_1=150$. For agent~2 bidding $b_2 = 150$ is unique. The expected reward to  agent~1 will be $r_1= 0.5\int_{X=100}^{200}X\frac{1}{100}dX=75$ ($0.5\$$ times the expected capacity). For agent~2, note that $b_2/2=75, 3b_2/2=225$, thus $r_2 = 0.5\int_{X=75}^{225}X\frac{1}{200}dX + 0.5\int_{X=225}^{250}225\frac{1}{200}dX \cong 70.3$. Therefore, if the grid takes both contracts, it pays about $\$145.3$ in expectation. The grid may also choose to assign a single contract and reduce the expense to $\$70-75$, however this means that the reduction goal of $M=150$ is only met w.p. of 50\% (regardless of which agent is selected).  

Next, suppose that the grid uses the DR-VCG mechanism. 
 Available contracts are cliff contracts with the SCE penalty scheme, i.e. for each $\ell$, a contract of the form $(\ell,f_\ell=\frac{\ell}{2},\alpha=\frac13,\beta=\frac12)$. 
Thus the expected penalty (and thus cost) for an agent $i$ on contract $j_\ell$ is $C^*_i(\ell) = \int_{X=0}^{\ell/3}\frac{\ell}{2}p_i(X) dX+  \int_{X=\ell/3}^{\ell}\frac{\ell-X}{2}p_i(X) dX$.
For example, for agent~1, 
$$C^*_1(\ell) =  \left\{\begin{array}{ll}
0 &, \ell \leq 100 \\
\frac{(\ell-100)^2}{200}  &, 100\leq \ell \leq 200 \\
\geq 50 &, 200 \leq \ell
\end{array} \right. $$
 
%

Suppose for simplicity that contracts are available only on multiples of 50~kW. Then agents report the following costs:
$$\begin{array}{l|llll}
\ell_i & 0 &  50  & 100 & 150  \\
\hline
C^*_1  & 0 &  0   & 0   & 12.5  \\
C^*_2  & 0 &  0   & 6.25  & 25  \\
\end{array}$$
Note that agent~2 reports higher costs, because he is less reliable and thus expects higher penalties. 
The cheapest combination of contracts for the grid is to assign 100 units to agent~1, and 50 units to agent~2. This will result in rewards of $r_1=C^*_2(150)-C^*_2(50) = 25-0=25$, and $r_2 =  C^*_1(150)-C^*_1(100) = 12.5-0=12.5$. Thus the total expense for the grid is $25+12.5=37.5$, which is less than what the DR-SCE mechanism pays a single agent.

In this example, the agents always meet the reduction goal. In the appendix, we show how payments and reliability vary under both mechanisms as we increase the reduction goal $M$.

\rmr{It should be easy to create an example where the improvement over the current mechanism is much more impressive}

\subsection{A Numeric Example expanded} \label{sec:example} \hma{I wasn't able to walk through the numbers...}

To demonstrate both our mechanism and its advantage over the current system, we consider a simple numeric example with only two agents (and no preparation costs), varying the reduction goal.  
\paragraph{Two agents with a linear penalty scheme}
Suppose there are two agents with the following capacity distributions at cost $0$:  $X_1\sim U[100,200]$ and $X_2\sim U[50,250]$. Thus both agents have the same capacity in expectation but agent~1 is more reliable.  

For the power company, $\lambda=0$ and $C_m=m$ for any $m$. Available contracts are cliff contracts of the form $(\ell,f_\ell=0,\alpha=0,\beta=1)$, which are simply linear penalties for hitting below the goal: $F(\ell,X) = [\ell-X]_+$. Thus the expected penalty (and thus cost) for an agent reporting $\ell_i$ is $C^*_i(\ell_i) =  \int_{X=0}^{\ell_i}(\ell_i-X)p_i(X) dX$.

For agent~1, 
$$C^*_1(\ell_1) =  \left\{\begin{array}{ll}
0 &, \ell_1 \leq 100 \\
\frac{(\ell_1-100)^2}{200}  &, 100\leq \ell_1 \leq 200 \\
\ell_1-150 &, 200 \leq \ell_1
\end{array} \right. $$

For agent~2, 
$$C^*_2(\ell_2) =  \left\{\begin{array}{ll}
0 &, \ell_2 \leq 50 \\
\frac{(\ell_2-50)^2}{400}  &, 50\leq \ell_2 \leq 250 \\
\ell_2-150 &, 250 \leq \ell_2
\end{array} \right. $$

Suppose for simplicity that agents only bid in multiples of 50. Then they report the following costs:
$$\begin{array}{l|llllllll}
\ell_i & 0 &  50  & 100 & 150 & 200   & 250 & 300 & 350\\
\hline
C^*_1  & 0 &  0   & 0   & 12.5 & 50   & 100 & 150 & 200\\
C^*_2  & 0 &  0   & 6.25  & 25 & 56.25 & 100 & 150 & 200\\
C_m    & 0  & 25  & 50    & 75 & 100   & 125 & 150 & 175 \\ 
\end{array}$$

We get the following solution for increasing values of $M$:
\begin{table}
$$\begin{array}{l|llllllll}
M &       50   & 100   & 150     & 200    & 250               & 300          &400  & 1000 \\
\hline
S(N,J) &  \{1\} & \{1\}& \{1,2\} & \{1,2\} &  \{1,2\}          & \{1,2\}     & \{1,2\} + 100 & \{1,2\} + 700\\
SC(N,J) & 0     & 0    &    0    &  6.25  &  18.75             &  37.5       & 87.5         & 387.5        \\  
\ell_1  &  50   & 100   & 100    &  100    &  150              & 150        & 150        &  150     \\ 
\ell_2  &  -    & -     & 50     &  100     & 100              &  150       & 150         & 150 \\
r_1 &     0     & 6.25 &  25 & 50-6.25   & 75-6.25          & 100-25     & 150 - 75 &    75  \\
   &            &      &         &  =43.75  & =68.75           & =75      &= 75   &  \\
r_2 &      -   &    -  &  12.5  &  37.5     &  62.5-12.5      & 87.5-12.5     & 137.5 - 62.5 & 75 \\
    &           &      &         &         &    =50           & =75           & = 75  & \\
r_1+r_2+C &  0   & 6.25  &  37.5   & 81.25  &  118.75           & 150     &  150+50 & 150+350\\ 
\hline
E[\text{grid expense}] & 0 & 6.25 & 37.5 & 75 & 100      &      112.5     &  162.5 & 462.5 \\   
\hline 
Pr(\sum X_i \geq M)        & 1  &1& 1    & 0.94   & 0.75        &  0.5    &  0.5      &  0.5 \\
Pr(\sum X_i \geq \frac34 M )& 1 &1& 1    & 1    & 0.96          &  0.86    & 0.94     &  1 \\  
Pr(\sum X_i \geq \frac12 M )& 1 &1& 1    & 1    & 1            &  1    & 1      &  1 \\  
\end{array}$$
\caption{\label{tab:linear}Each column shows the results for a different reduction goal between $50$ and $1000$. The rows show: the selected set (+additional production); the social cost; the individual reduction commitments; the rewards to selected agents; the total grid expense (total rewards + external cost - expected penalties); and the probability to meet the reduction goal or part of it. }
\end{table}
The grid expense is the total reward paid, plus the cost of the outside option, minus the penalties paid by the agents (note that since is no preparation cost, the sum of expected penalties equals the social cost).
The success probabilities are computed under the assumption that $X_i$ are independent.  

We can see that the best thing for the grid is to let each agent commit to up to 150 units, where more units go the more reliable agent~1. When this is exhausted the grid uses its outside option to complete the missing power. We can see that the total cost for the grid (rewards + outside option) is always at most $C_M$, and substantially lower for small $M$. 
What if use the current penalty scheme rather than linear penalties? since the agents never commit to more than $\ell_i=150$, and $X_o \geq 50\geq \ell_i/3$, the current penalty scheme would never hit the `cliff' and hence for this instance it would give identical outcome. 
 
In contrast, suppose that the grid uses the DR-SCE mechanism. The dominant strategy for agent~1 is to bid anything s.t. $[100,200]\subseteq [b_1/2,3b_1/2]$, thus say $b_1=150$, and for agent~2 $b_2 = 150$ is unique. The expected reward to each selected agent will be $r_1=75$ ($0.5\$$ times the expected capacity) and $r_2\equiv 70$. Suppose the grid selects agents at random until it meets the reduction goal $M$ and only then uses its outside option.\footnote{The current system in fact recruits all agents regardless of the reduction goal so is even less efficient.} 
\begin{table}
$$\begin{array}{l|llllllll}
M &                     50    &     100     & 150  &   200    & 250      & 300          &400            & 1000 \\
\hline
|S| &                 1    &       1      &  1    &  2       & 2   &   2        &   2 (+100)        & 2 (+ 700) \\
E[\text{grid expense}] 
                     &  72.5     &   72.5 &  72.5 &  145      & 145  &    145     & 195         & 495  \\ 
\hline
Pr(\sum X_i \geq M)     & 1      & 0.875  &  0.5    &   0.94  &  0.75    & 0.5          & 0.5             & 0.5              \\
Pr(\sum X_i \geq \frac34 M )&  1 &  0.937 &  0.78    &     1  &  0.96   & 0.86       &     0.94              &  1   \\  
Pr(\sum X_i \geq \frac12 M )&  1 &  1    &  0.937   &     1   &   1   &  1        &         1        &      1 \\  
\end{array}$$
\caption{\label{tab:current} Outcome of the current PG\&E mechanism for various reduction goals.}
\end{table}
We can see that the current mechanism is always at least as expensive to the grid as using our proposed mechanism with the current penalty scheme, and the difference is substantial for $M\leq 400$. For $M\leq 150$ the current mechanism is also less reliable, since it relies on a single agent while our mechanism signs contracts with both agents (and at a lower cost). Note that this improvement is attained without modifying or optimizing the penalty scheme.

When $M$ is very high then both mechanisms rely mainly on the outside option and thus there is not much difference between them.

\rmr{It should be easy to create an example where the improvement over the current mechanism is much more impressive}

\end{document}